\documentclass[10pt,journal]{IEEEtran}

\usepackage{basicreq}

\IEEEoverridecommandlockouts
\title{Improving the Performance of Nested Lattice Codes Using Concatenation}
\author{Shashank Vatedka,~\IEEEmembership{Student~Member,~IEEE,}
Navin~Kashyap,~\IEEEmembership{Senior~Member,~IEEE}
\thanks{This work was presented in part at the 2016 National Conference on Communications, Guwahati, Assam, India.}
\thanks{S.~Vatedka and N.~Kashyap are with the Department of Electrical Communication Engineering, Indian Institute of Science, Bangalore, India. Email: \texttt{\{shashank,nkashyap\}@ece.iisc.ernet.in} }
}

\begin{document}
\maketitle
\begin{abstract}
A fundamental problem in coding theory is the design of an efficient coding scheme that achieves the capacity of the additive white Gaussian (AWGN) channel.
%There has been considerable progress in this regard, and we now have capacity-achieving schemes with low decoding  complexity.
The main objective of this short note is to point out that by concatenating a capacity-achieving nested lattice code with a suitable high-rate linear code over an appropriate finite field,
we can achieve the capacity of the AWGN channel with polynomial encoding and decoding complexity. Specifically, we show that using  inner Construction-A lattice codes and outer Reed-Solomon codes, we can obtain capacity-achieving codes whose encoding and decoding complexities grow as $O(N^2)$, while the probability of error decays exponentially in $N$, where $N$ denotes
the blocklength. Replacing the outer Reed-Solomon code by an expander code helps us further reduce the decoding complexity to $O(N\log^2N)$. This also gives us a recipe for converting a high-complexity nested lattice code for a Gaussian channel to a low-complexity concatenated code without any loss in the asymptotic rate. As examples, we describe polynomial-time coding schemes for the wiretap channel, and the compute-and-forward scheme for computing integer linear combinations of messages.
\end{abstract}

\section{Introduction}\label{sec:intro}

The problem of designing efficient coding schemes for the additive white Gaussian (AWGN) channel has been studied for a very long time.
Shannon~\cite{Shannon} showed that random codes can achieve the capacity of the AWGN channel. Showing that structured codes can achieve capacity remained open till it was shown in~\cite{deBuda,LinderCorrect}, and later in~\cite{Urbankelattice},
that lattice codes can achieve capacity with maximum likelihood (ML) decoding. Erez and Zamir~\cite{Erez04} then showed that nested lattice codes
can achieve capacity with closest lattice point decoding. 
Lattice codes have been shown to be optimal for several other problems such as dirty paper coding, Gaussian multiple access channels, quantization, and so on.
 They have also been used in the context of physical layer network coding~\cite{LiewPNC,Nazer11} and physical layer security~\cite{LingWiretap,Vatedka15}.
 We refer the reader to the book by Zamir~\cite{zamirbook} for an overview of the applications of lattices for channel coding and quantization.
A drawback with the proposed nested lattice schemes is that there are no known polynomial-time algorithms for encoding and decoding.
% Unlike the case of linear codes for discrete memoryless channels, the encoding complexity of lattice codes also grows exponentially with $N$. 
A notable exception is the polar lattice scheme~\cite{YanPolar} which can achieve the capacity of the AWGN channel with an encoding/decoding complexity of $O(N\log^2N)$.\footnote{Yan et al.~\cite{YanPolar} also show that for a fixed error probability (as opposed to a probability of error that goes to zero as $N\to\infty$), the encoding/decoding complexity of polar lattices is $O(N\log N)$.} However, the probability of error goes to zero as $e^{-\Omega(N^{\beta})}$ for any $0<\beta<1/2$. 

There are also lattice constructions with low decoding complexity~\cite{Sadeghi_ldpclattice,Sakzad_turbolattice,Sommer} that have been empirically shown to achieve rates close to capacity. However, it is still an open problem to theoretically show that these codes achieve the capacity of the AWGN channel. Low density Construction-A (LDA) lattices~\cite{diPietro_arxiv} are a class of lattices obtained from low-density parity-check codes and have been shown to achieve capacity with lattice decoding. Simulation results suggest that they can approach capacity with low-complexity belief propagation decoding, but we still do not have a theoretical proof of the same.

Concatenated codes were introduced by Forney~\cite{Forney_concat} as a technique for obtaining low-complexity codes that can achieve 
the capacity of discrete memoryless channels. Concatenating an inner random linear code with an outer Reed-Solomon code is a simple way of designing good codes.
Using this idea, Joseph and Barron~\cite{JosephSuperposition} proposed the capacity-achieving sparse regression codes for the AWGN channel, having quadratic (in the blocklength $N$) encoding/decoding complexity. They used a concatenated coding scheme with an inner sparse superposition code and an outer Reed-Solomon code. The probability of decoding error goes to zero exponentially in $N/\log N$~\cite{JosephExponential}. Recently,~\cite{Rush_SPARC_AMP} proposed an approximate message passing scheme  having complexity that grows roughly as $O(N^2)$ for decoding sparse regression codes, and showed that the new decoder guarantees a vanishingly small error probability as $N\to\infty$ for all rates less than the capacity. However,~\cite{Rush_SPARC_AMP} does not provide any guarantees for the rate of decay of the probability of error.

The objective of this article is to show that using the technique of concatenation, we can reduce the asymptotic decoding complexity of nested lattice codes that operate at rates close to capacity. 
We start with a sequence of nested lattice codes having rate $C-\delta$, where $C$ denotes the capacity of the AWGN channel, and $\delta$ is a small positive constant that denotes the gap to capacity. These codes typically have exponential encoding/decoding complexity. By concatenating these with suitable linear codes, we obtain a sequence of concatenated codes that have transmission rate at least $C-2\delta$, but whose encoding/decoding complexity scales polynomially in the blocklength. 
We show that concatenating an inner nested lattice code with an outer Reed-Solomon code yields a capacity-achieving coding scheme whose encoding/decoding complexity is quadratic in the blocklength. Furthermore, the probability of error decays exponentially in $N$. Replacing the Reed-Solomon code with an expander code~\cite{Zemorconcatenated} yields a capacity-achieving coding scheme with decoding complexity $O(N\log^2N)$ and encoding complexity $O(N^2)$. To the best of our knowledge, this is the first capacity-achieving coding scheme for the AWGN channel whose encoding and decoding complexities are polynomial, and the probability of error decays exponentially in the blocklength.
The techniques that we use are not new, and we use results of~\cite{Forney_concat} and~\cite{Erez04} to prove our results.
An attractive feature of this technique is that it can also be used to reduce the complexity of nested lattice codes for several other Gaussian networks.
It can be used as a tool to convert any nested lattice code having exponential decoding complexity to a code having polynomial decoding complexity.
This comes at the expense of a minor reduction in performance (in terms of error probability) of the resulting code. 
Furthermore, we are able to give guarantees only for large blocklengths.
% However this technique can be extended to obtain low-complexity schemes for other Gaussian channels.
As applications, we show how these ideas can be used to obtain a capacity-achieving scheme for the Gaussian wiretap channel and to reduce the decoding complexity of the compute-and-forward protocol for Gaussian networks.
More recently, these techniques have also been used to obtain polynomial-time lattice coding schemes for secret key generation from correlated sources~\cite{Vatedka_skgeneration}.

Throughout this article, we measure complexity in terms of the number of binary operations required for decoding/encoding, and
we are interested in how this complexity scales with the blocklength. 
We assume that arithmetic operations on real numbers are performed using floating-point arithmetic, and that each real number has a  $t$-bit binary representation, with $t$ being independent of the blocklength. The value of $t$ would depend on the computer architecture used for computations (typically 32 or 64 bits). In essence, we assume that each floating-point operation has complexity $O(1)$.

The rest of the paper is organized as follows. We describe the notation used in the paper and recall some concepts related to lattices in Section~\ref{sec:notation}. We then describe the concatenated coding scheme for the AWGN channel in Section~\ref{sec:codingscheme}, with Theorem~\ref{thm:main_RS}
summarizing the main result. 
In Section~\ref{sec:parallelconcatenation}, we use an outer expander code to reduce the decoding complexity to $O(N\log^2N)$. 
This is summarized by Theorem~\ref{thm:main_expander}.
The performance of the two concatenated coding schemes are compared with polar lattices and sparse superposition codes in Table~\ref{table:comparison}.
We make some remarks on extending these ideas to the Gaussian wiretap channel and the compute-and-forward protocol in Section~\ref{sec:discussion}. We also indicate how the same technique can be used to reduce decoding complexity and improve the probability of error of LDA lattices and polar lattices. We conclude the paper with some final remarks in Section~\ref{sec:concluding}. The proof of Lemma~\ref{lemma:CF_lowcomplexity} is provided in Appendix A.

\section{Notation and Definitions}\label{sec:notation}
% We borrow notation from \cite{Vatedka15}. 
For a detailed exposition on lattices and their applications in several communication-theoretic problems, see~\cite{zamirbook}. We denote the set of integers by $\Z$ and real numbers by $\R$. The set of nonnegative real numbers is denoted by $\R^+$.
For a prime number $p$ and positive integer $k$, we let $\Fpk$ denote the field of characteristic $p$ containing $p^k$ elements. 
% If $X$ and $Y$
% are random variables, then $I(X;Y)$ denotes the mutual information between $X$ and $Y$. 
For $A,B\subset \R^n$ and $a,b\in \R$, we define $aA+bB$ to be the set $\{ ax+by:x\in A,y\in B \}$.
Given $\x,\y\in \Z^n$ and $p\in \Z$, we say that $\x\equiv \y (\bmod p)$ if $(\x-\y)\in p\Z^n$. We use the standard big-O and little-O notation to express the asymptotic relationships between various quantities.

If $G$ is an $n\times n$ full-rank matrix with real entries, then the set $\L = G^T\Z^n\triangleq \{ G^T\z : \z\in\Z^n \}$ is called a \emph{lattice} in $\R^n$.
We say that $G$ is a \emph{generator matrix} for $\L$. 
% It is a fact that the generator matrix of a lattice is not unique.
Let $Q_{\L}(\cdot)$ denote the lattice \emph{quantizer} that maps a point in $\R^n$ to the point in $\L$ closest to it. For $\x\in \R^n$, we define $[\x]\bmod \L$ to be the quantization error $\x-Q_{\L}(\x)$ when using the quantizer $Q_{\L}$. The \emph{fundamental Voronoi region} of $\L$, $\cV(\L)$, is defined to be $\cV(\L)\triangleq \{\x\in \R^n:Q_{\L}(\x)=\0\}$. The radius of the smallest closed ball in $\R^n$ centered at zero which contains $\cV(\L)$ is called the \emph{covering radius}, and is denoted $\rcov(\L)$.
 Given two lattices $\L$ and $\Lc$ in $\R^n$, we say that $\Lc$ is \emph{nested} in $\L$ if $\Lc\subset \L$. We call $\Lc$ the \emph{coarse lattice} and $\L$ the \emph{fine lattice}.

\subsection{Construction A}\label{sec:constructionA}
For completeness, we describe Construction A, a technique to obtain lattices from linear codes over prime fields --- see~\cite{Erez04,zamirbook} for a more detailed description.
Let $p$ be a prime number and $\cC$ be an $(n,k)$ linear code over $\bFp$, i.e., $\cC$ has length $n$ and dimension $k$. Then, the Construction-A
lattice obtained from $\cC$, denoted $\L_\mathrm{A}(\cC)$, is defined as the set of all points $\x$ in $\Z^n$ such that $\x\equiv \y\: (\bmod p)$ for some $\y\in\cC$.
Note that $p\Z^n$ is always a sublattice of $\L_\mathrm{A}(\cC)$.

We will use the nested lattice construction from~\cite{Erez04}. Let $\Lc$ be a (possibly scaled) Construction-A lattice in $\R^n$, having a generator matrix $G$. Let $\L_\mathrm{A}(\cC)$ be another Construction-A lattice obtained from an $(n,k)$ linear code $\cC$ over $\bFp$. Then, $\L\triangleq \frac{1}{p}G^T\L_\mathrm{A}(\cC)=\{ (1/p)G^T\x : \x \in \L_\mathrm{A}(\cC) \}$ is also a lattice, and it can be verified that $\Lc$ is nested in $\L$. We will refer to $(\L,\Lc)$ as a \emph{nested Construction-A} lattice pair. A key feature of this nested lattice pair that will be of use to us is that $\L\cap \cV(\Lc)$ (and hence the quotient group $\L/\Lc$) contains $p^k$ elements. Furthermore, there exists a group isomorphism between the quotient group $\L/\Lc$ and $\Fpk$, the latter being viewed as an additive group.

\section{Coding Scheme for the AWGN Channel}\label{sec:codingscheme}
%\begin{figure}
%\begin{center}
%\resizebox{7cm}{!}{\input{blockdiagram.pdf_t}}
%\end{center}
%\end{figure}
Let us consider the point-to-point AWGN channel where the source encodes its message $M$ to $\u\in\R^n$ and transmits this to a destination that receives 
\[
	\w = \u + \z,
\]
where $\z$ is the noise vector having independent and identically distributed (iid) Gaussian entries with mean zero and variance $\nsvar$. Erez and Zamir~\cite{Erez04} proposed a capacity-achieving nested lattice
scheme for the AWGN channel, which we briefly describe here. The code is constructed using a pair of nested lattices $(\Lfn,\Lcn)$, where $\Lcn\subset\Lfn\subset\R^n$.
The codebook consists of all the points of $\Lfn$ within the fundamental Voronoi region of $\Lcn$, i.e., the codebook is $\Lfn\cap\cV(\Lcn)$.
The transmission rate is therefore $\frac{1}{n}\log_2|\Lfn\cap\cV(\Lcn)|$.

The source also generates a random dither vector $\t$, uniformly distributed over $\cV(\Lcn)$, which is assumed to be known to the decoder\footnote{In principle, the dither vector is not necessary (see, e.g.,~\cite[Section IV]{diPietro_arxiv}). However, this technique of dithered transmission simplifies the analysis of the probability of error of the decoder.}. 
Each message $M$ is mapped to a point $\x$ in the codebook $\Lfn\cap\cV(\Lcn)$. The encoder $\cEn$ takes the message $M$ as input, and outputs the 
vector $[\x-\t]\bmod\Lcn$, which is transmitted across the channel. This process of translating the message by $\t$ modulo $\Lcn$ prior to transmission is called \emph{dithering}. The encoder satisfies a maximum transmit power constraint given by $\frac{1}{n}\max_{\u\in\cV(\Lc) } \Vert \u \Vert^2=\frac{1}{n}\rcov^2(\Lcn)<P$. 

Upon receiving $\w$, the receiver uses a decoder $\cDn$ to estimate $M$, which does the following. It computes $\widetilde{\w}=[\alpha \w+\t]\bmod\Lcn$, where
$\alpha = \frac{P}{P+\nsvar}$. The estimate of $M$ is the message that corresponds to $[Q_{\Lfn}(\widetilde{\w})]\bmod\Lcn$.

Let $C\triangleq  \frac{1}{2}\log_2\left( 1+\frac{P}{\nsvar}  \right)$.
Erez and Zamir~\cite{Erez04} showed that there exist nested lattices with which we can approach the capacity of the AWGN channel. Specifically,
\begin{lemma}[\cite{Erez04}, Theorem 5]
	For every $\epsilon>0$, there exists a sequence of nested Construction-A lattice pairs $(\Lfn,\Lcn)$ 
	%with $\Lcn\subset\Lfn$ and associated encoder-decoder pairs $(\cEn,\cDn)$ 
	such that for all sufficiently large $n$, the maximum transmit power is
	\[
		\frac{1}{n}\rcov^2(\Lcn) \leq  P+\epsilon,
	\] 
	the transmission rate is
% 	\footnote{\textup{The subscript `in' in $\Rin^{(n)}$ indicates that we intend to use the nested lattice code as an inner code in a concatenated coding scheme, which we describe in Section~\ref{sec:concat_scheme_awgn}.}}
	\[
		R^{(n)} \triangleq  \frac{1}{n}\log_2|\Lfn\cap\cV(\Lcn)| \geq C-\epsilon,
	\]
	and the probability of error decays exponentially in $n$ for all $R^{(n)}<C$, i.e., there exists a function $E:\R^+ \to \R^+$ so that for every $\x\in \Lfn\cap\cV(\Lcn)$ and all sufficiently large $n$, we have
	\[
		\Pr[ \x \neq \cDn(\cEn(\x)+\z) ] \leq e^{-n E(R^{(n)})}.
	\]
	Furthermore, the quantity $E(R^{(n)})$ is positive for all $R^{(n)}<C$.
	\label{lemma:ErezResult}
\end{lemma}
The decoding involves solving two closest lattice point problems, which are the $Q_{\Lfn}$ and $\bmod \Lcn$ operations.
Therefore, the decoding complexity is $O(2^{nR^{(n)}})$. If the encoder uses a look-up table to map messages to codewords, the complexity would also be $O(2^{R^{(n)}})$.

\subsection{The Concatenated Coding Scheme for the AWGN Channel}\label{sec:concat_scheme_awgn}

\begin{figure}
\begin{center}
\resizebox{8cm}{!}{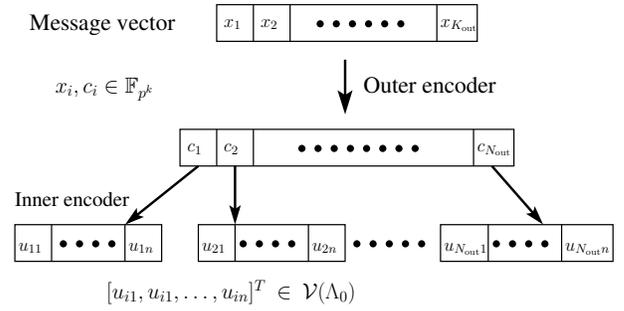}
\end{center}
\caption{Illustration of the concatenated nested lattice coding scheme.}
\label{fig:concatenatedcode}
\end{figure}
Let us now give a brief description of the concatenated coding scheme. See~\cite{Forney_concat} for a more detailed exposition and application to the discrete memoryless channel. The code has two components: 
\begin{itemize}
\item Inner code: A nested Construction-A lattice code $(\Lfn,\Lcn)$ with the fine lattice $\Lfn$ obtained from a $(n,k)$ linear code over $\bFp$. 
\item Outer code: An $(\Nout,\Kout,\dout)$ linear block code (where $\dout$ is the minimum distance of the code) over $\Fpk$.
\end{itemize}

The message set has size $ p^{k\Kout}$, and each message can be represented by a vector in $\Fpk^{\Kout}$. The outer code maps this vector to a codeword in $\Fpk^{\Nout}$ in a bijective manner. Let us call this $\mathbf{c}_{\mathrm{out}}=[c_1\:c_2\:\cdots c_{\Nout}]^T$, where each $c_i\in\Fpk$. The inner code maps each $c_i\in\Fpk$ to a point in $\Lfn\cap\cV(\Lcn)$. This results in a codeword of length $n\Nout$ having real-valued components. Each inner codeword is dithered by an independent dither vector prior to transmission. The encoding process is illustrated in Fig.~\ref{fig:concatenatedcode}.
The receiver first uses the decoder for the inner code to estimate the components $c_i$, and finally uses the decoder for the outer code to recover the message. 
Since the outer code has minimum distance $\dout$, the message is guaranteed to be recovered correctly if not more than $(\dout-1)/2$ inner codewords are in error. 
Furthermore, if all the inner codewords satisfy the (max) power constraint, then the concatenated code is also guaranteed to satisfy the same.

We now show that using this technique, we can achieve the capacity of the AWGN channel.
Let us fix $\delta>0$.
Suppose we choose a sequence of nested lattice codes $(\Lfn,\Lcn)$ that are guaranteed by Lemma~\ref{lemma:ErezResult}.
Let $\Rin=\frac{1}{n}\log_2|\Lfn\cap\cV(\Lcn)|$ denote the  rate of the nested lattice code\footnote{While $\Rin$ depends on $n$, we have chosen not to include this dependence in the notation, to avoid clutter.}.
Recall that the number of cosets, $|\Lfn\cap\cV(\Lcn)|$, is equal to $p^k=2^{n\Rin}$. 
For every $n$, let us concatenate the nested lattice code with an outer $(\Nout,\Kout,\Nout-\Kout+1)$ Reed-Solomon code over $\Fpk$, where $\Nout = p^k-1$ and 
\begin{equation}
\Kout = \Nout (1-2e^{-nE(\Rin)}-2\delta).
\label{eq:kout}
\end{equation} 
% Here and in the rest of this section, we drop the superscript in $\Rin^{(n)}$ for convenience.
The resulting code, which we denote $\cCn$, has blocklength $N=n\Nout\approx n2^{n\Rin}$, and rate 
\begin{equation}
R^{(N)}=\frac{\Kout}{n\Nout}\log_2 p^k=(1-2e^{-nE(\Rin)}-2\delta)\Rin.
\label{eq:effectiverate}
\end{equation}

\begin{theorem}
For every $\epsilon>0$, there exists a  sequence of concatenated codes $\cCn$ with inner nested lattice codes and outer Reed-Solomon codes that satisfies the following for all sufficiently large $n$:
\begin{itemize}
\item rate $R^{(N)}\geq C-\epsilon$, 
\item maximum transmit power $$\max_{\x\in \cCn}\frac{1}{N}\Vert \x \Vert^2\leq P-\epsilon,$$
\item the probability of error is at most $e^{-NE(\Rin)\epsilon}$, and
\item the  encoding and decoding complexities grow as $O(N^2)$.
\end{itemize}
\label{thm:main_RS}
\end{theorem}
\begin{proof}
The construction of the concatenated codes ensures that the power constraint can be satisfied. 
From Lemma~\ref{lemma:ErezResult}, we are assured of a nested lattice code such that $\Rin$ is at least $C-\epsilon/2$.
Choosing a small enough $\delta$ and a large enough $n$ in (\ref{eq:effectiverate}) guarantees that the 
 rate of the concatenated code, $R^{(N)}$, is at least $C-\epsilon$ for all sufficiently large $N$. 

Let us now proceed to analyze the probability of error. Clearly, the probability that an inner codeword is in error is upper bounded by $e^{-nE(\Rin)}$ by Lemma~\ref{lemma:ErezResult}. Since the outer Reed-Solomon code has minimum distance $\Nout-\Kout+1$, the decoder makes an error only if at least $(\Nout-\Kout+1)/2 = \Nout(e^{-nE(\Rin)}+\delta+1/(2\Nout))$ inner codewords are in error. For all sufficiently large $N$,
we can upper bound the probability of decoding error as follows:
\begin{align}
 P_e^{(N)} &\leq \begin{pmatrix}
 				  \Nout \\\Nout(e^{-nE(\Rin)}+\delta+1/(2\Nout))	
                \end{pmatrix} &\notag\\
                &\qquad \times\left(e^{-nE(\Rin)}\right)^{\Nout(e^{-nE(\Rin)}+\delta+1/(2\Nout))} &\label{eq:1_concat}\\
           &\leq \begin{pmatrix}
            				  \Nout \\\Nout(e^{-nE(\Rin)}+2\delta)	
                           \end{pmatrix}  &\notag\\
                           &\qquad \times\left(e^{-nE(\Rin)}\right)^{\Nout(e^{-nE(\Rin)}+\delta)} &\label{eq:2_concat}\\
           &\leq e^{\Nout h(e^{-nE(\Rin)}+2\delta)} \left(e^{-nE(\Rin)}\right)^{\Nout(e^{-nE(\Rin)}+\delta)} \label{eq:3_concat}
\end{align}
where (\ref{eq:1_concat}) is obtained using the union bound, and the last step from Stirling's formula. In (\ref{eq:3_concat}), $h(\cdot)$ denotes the binary entropy function.
For all sufficiently large $n$, we have $h(e^{-nE(\Rin)}+2\delta)<h(3\delta)$. Using this in the above and simplifying, we get
\begin{align}
P_e^{(N)} & \leq \text{exp}\Big( -n\Nout\big(E(\Rin)(e^{-nE(\Rin)}+\delta) - h(3\delta)/n\big) \Big) &\notag
\end{align} 
Let us define the error exponent as 
\begin{equation}
 E_{\textrm{conc}} \triangleq E(\Rin)(e^{-nE(\Rin)}+\delta) - h(3\delta)/n
\end{equation}
It is clear that $E_{\textrm{conc}}>E(\Rin)\delta/2$ for all sufficiently large $n$. This proves that the probability of error decays exponentially in $N$.

Let us now inspect the encoding and decoding complexity. Table~\ref{table:summary_parameters} summarizes the relationships between the various parameters. As remarked in the introduction, we assume that each floating-point operation requires a constant number of binary operations (i.e., independent of $N$) and has a complexity of $O(1)$. Encoding/decoding each inner (nested lattice) codeword requires  $O(2^{n\Rin})$ floating-point operations, and there are $\Nout$ many codewords, leading to a total complexity of $O(N^2)$. Furthermore, encoding/decoding a Reed-Solomon codeword requires $O(\Nout^2)$ operations in $\Fpk$~\cite[Chapter~6]{RothBook}. 
Multiplication and inversion are the most computationally intensive operations in $\Fpk$, and they can be performed using $O((k\log_2 p)^2)=O(n^2)$ binary operations~\cite[Chapter~2]{Hankerson_FFieldComplexity}. Therefore, the outer code has an encoding/decoding complexity of $O(\Nout^2)\times O(n^2) = O(N^2)$. We can therefore conclude that encoding and decoding the concatenated code requires a complexity of $O(N^2)$.
\end{proof}

\subsection{Complexity}
Let us denote $\chi$ to be the decoding complexity. From Theorem~\ref{thm:main_RS}, we can conclude that for a fixed gap to capacity ($\gamma\triangleq C-R$), the probability of error for the concatenated coding scheme scales as $e^{-\Omega(\sqrt{\chi})}$.
As argued in~\cite[Section~5.1]{Forney_concat}, this is a much stronger statement than saying that the decoding complexity is polynomial in the blocklength.
%We now show the following result, which says that for a fixed probability of error, the decoding complexity depends polynomially on the gap to capacity:
%\begin{lemma}
%Let $\Delta\triangleq C-R$. For a fixed probability of error and in the regime of small $\Delta$, we have $N=O(1/\Delta^3)$, and the decoding/encoding complexity scales as $O(1/\Delta^6)$.
%\label{lemma:gaptocapacity}
%\end{lemma}
%\begin{proof}
%Let us choose $\Rin=C-\Delta/2$. From Theorem~\ref{thm:main_RS}, the probability of error of the concatenated coding scheme can be upper bounded by $e^{-N(E(\Rin)\Delta)}$ for all sufficiently large $N$.
%From~\cite[Theorem 5]{Erez04}, we have, for all sufficiently small $\Delta$, 
%\[
%	E(\Rin) = a_1 [e^{2\Delta}-1 - 2\Delta],
%\]
%where $a_1$ is a positive constant. It can be easily verified that $e^{2\Delta}-1 - 2\Delta = 2\Delta^2(1-o_\Delta(1))$,
%where $o_\Delta(1)$ is a quantity that vanishes as $\Delta\to 0$. Therefore, the probability of error,
%\[
% P_e^{(N)} \leq e^{-2a_1 N\Delta^3(1-o_\Delta(1))}
%\]
%for all sufficiently large $N$. We can therefore deduce that for a fixed probability of error, $N$ must scale as $O(1/\Delta^3)$.
%Since the encoding and decoding complexities grow as $O(N^2)$, we can conclude that this is proportional to $O(1/\Delta^6)$.
%\end{proof}
%\subsubsection{Previous Work} 

Previously, Joseph and Barron~\cite{JosephSuperposition,JosephExponential} proposed a concatenated coding scheme with inner sparse superposition codes and outer Reed-Solomon codes. They showed that their scheme achieves the capacity of the AWGN channel with polynomial (in the blocklength $N$) time  encoding/decoding.
The decoding complexity is $\chi=O(N^2)$.
However, the probability of error decays exponentially in $N/\log N$ for a fixed gap to capacity $\gamma$. Therefore, the probability of error is exponentially decaying in  $\sqrt{\chi}/\log \chi$.  More recently, Yan et al.~\cite{YanPolar} proposed a lattice-based scheme using polar codes that achieves capacity with an encoding/decoding complexity of $\chi=O(N\log^2N)$. The probability of error (for a fixed $\gamma$) is  $e^{-\Omega(N^{\beta})}$, for any $0<\beta<0.5$. The probability of error is therefore $e^{-\Omega((\chi/\log^2 \chi)^{\beta})}$. The concatenated scheme we have studied here outperforms these works in the sense that the probability of error decays exponentially in the square root of $\chi$ for a fixed $\gamma$.
However, we have not been able to show that for a fixed probability of error, the decoding complexity is polynomial in the gap to capacity, i.e., $\chi=O(\gamma^{-a})$ for some positive constant $a$.
The only such result for Gaussian channels that we are aware of is by Yan et al.~\cite{YanPolar}, where they showed that polar lattices have a decoding complexity that is polynomial in the gap to the Poltyrev capacity (for the AWGN channel without restrictions/power constraint). Finding a capacity-achieving coding scheme for the power-constrained AWGN channel with a decoding complexity that scales polynomially in the gap to capacity for a fixed probability of error still remains an open problem.

\section{Reduced Decoding Complexity using Expander Codes}\label{sec:parallelconcatenation}

\begin{table*}
   \begin{tabular}{|m{0.15\textwidth}|m{0.15\textwidth}|m{0.15\textwidth}|m{0.2\textwidth}|m{0.2\textwidth}|}
      \hline
      {\bf Concatenation scheme} & {\bf Relation between $n$ and $\Nout$} & {\bf Rate of inner code}  &  {\bf Rate of outer code $R_{\mathrm{out}}=\Kout/\Nout$} & {\bf Relation between $n,\Nout$ and overall blocklength $N$}   \\\hline
      {\bf Reed-Solomon} & $\Nout = 2^{n\Rin}-1$ & $\Rin\geq C-\epsilon/2$ &  $1-2e^{-nE(\Rin)}-2\delta$ & $ N=2^{\Theta(n)} $ \\ 
                         &                       &                      &                             &  $N= \Theta(\Nout\log\Nout)$ \\ \hline
       {\bf Expander}	 & $n$ a fixed but sufficiently large constant independent of $\Nout$. & $\Rin\geq C-\epsilon/2$ & $R_{\mathrm{out}}\geq 1-2\left( 4\sqrt{\epsilon}-\frac{1}{\Delta}\right)$ & $n=\Theta_{\Nout}(1)$ \\
			 & 			&			&			      & $N=\Theta(\Nout)$\\\hline 
   \end{tabular}
   \smallskip
   \caption{Summary of the various parameters. Note that the overall blocklength is $N=n\Nout$, and the rate of the concatenated code is $\Rin R_{\mathrm{out}}$.}
   \label{table:summary_parameters}
\end{table*}

In this section, we present a concatenation scheme that reduces the decoding complexity to $O(N\log^2N)$. This is based on the parallel concatenation approach~\cite{BargZemor_concatenated}
of using outer expander-type codes to obtain ``good'' linear codes.
%We will replace the outer Reed Solomon code used in the previous section by an expander code to reduce the decoding complexity. 
It was shown in~\cite{BargZemor_concatenated} that for binary channels, parallel concatenation yields an error performance similar to that of serial concatenation (concatenation using a Reed-Solomon code), but reduces the decoding complexity.

\subsection{The Coding Scheme}\label{sec:expander_scheme}

Let us fix an $\epsilon>0$ (where $\epsilon\ll 1$), and let our target rate be $R=C-\epsilon$, where $C$ denotes the capacity of the AWGN channel. 
As in the previous section, let $\Nout$ and $n$ denote the blocklengths of the outer and inner codes respectively. The overall blocklength is $N=n\Nout$. 
Unlike the previous section, however, we fix $n$ to be a sufficiently large constant, and we let $\Nout$ grow to infinity.

Let us fix the rate of the inner code to be
\[
		\Rin^{(n)} = C-\frac{\epsilon}{2},
	\]
	Then, Lemma~\ref{lemma:ErezResult} guarantees the existence of a sequence of nested Construction-A lattice codes $(\Lfn,\Lcn)$ with rate $ \frac{1}{n}\log_2|\Lfn\cap\cV(\Lcn)| \geq \Rin^{(n)}$, for which the probability of error, 
	$\peinner\triangleq  \Pr\Big[ \x \neq \cDn(\cEn(\x)+\z) \Big]$, is at most $e^{-n E(\Rin^{(n)})}$
   for all sufficiently large $n$. Furthermore, these lattices are obtained from linear codes over $\mathbb{F}_p$ for prime $p$ (which is a function of $n$). 
   Let us fix an $n$ large enough so that 
   \begin{equation}
      -\ln\peinner = nE(\Rin^{(n)}) \geq \frac{ h\left( 4\sqrt{\epsilon}\right)  +\epsilon^2  }{0.8\epsilon}, 
      \label{eq:peinner_condition}
   \end{equation}
where $h(\cdot)$ denotes the binary entropy function. For a fixed $\epsilon$, the parameters $n$ and $p$ will remain constant, and we will let only $\Nout$ grow to infinity.

\subsubsection{The outer code}\label{sec:expander_outercode}
We use an outer expander code whose construction is similar to the one in~\cite{Zemorconcatenated}. This has two components:
\begin{itemize}
   \item A sequence of $\Delta$-regular bipartite expander graphs $\cGN=(\cAN,\cBN,\cEN)$ with vertex set $\cAN\cup\cBN$ and edge set $\cEN$, with $|\cEN|=\Nout$. Here, $\cAN$ denotes the set of left vertices and $\cBN$ denotes the set of right vertices, with $|\cAN|=|\cBN|=\frac{\Nout}{2\Delta}$, where $\Delta$ is a  large constant independent of $\Nout$. The graph $\cGN$ is chosen so that the second-largest eigenvalue of its adjacency matrix, denoted $\lambda(\cGN)$, is at most $2\sqrt{\Delta-1}$. Explicit constructions of such graphs can be found in the literature~\cite{Ramanujangraph} (see~\cite{Marcus} for a stronger result). This graph is a normal factor graph for the outer expander code.
   We choose a sufficiently large $\Delta$ so that the inequality $\frac{2\sqrt{\Delta-1}}{\Delta}\leq \sqrt{\epsilon}$ holds.
   \item A linear code $\cC_0$ over $\Fpk$ having blocklength $\Delta$ and dimension $k_0$. For convenience, let us choose $\cC_0$ to be a $(\Delta,k_0)$ Reed Solomon code over $\Fpk$ (assuming that $\Delta<p^k$) with $k_0=\Delta\left(1-4\sqrt{\epsilon}\right)+1$. The minimum Hamming distance of $\cC_0$ is $d_0=\Delta-k_0+1=4\sqrt{\epsilon}\Delta$. Let us define 
   \begin{equation}
    \delta_0\triangleq  \frac{d_0}{\Delta}=4\sqrt{\epsilon}.
    \label{eq:delta_0}
   \end{equation}
\end{itemize}
Let us order the edges of $\cGN$ in any arbitrary fashion, and for any $v\in \cAN\cup\cBN$, let $E_v\triangleq \{ e_v(1),\ldots , e_v(\Delta) \}$ denote the set of edges incident on $v$, where $e_v(1)<e_v(2)<\cdots <e_v(\Delta)$ according to the order we have fixed. We define the expander code as follows: The codeword entries are indexed by the edges of $\cGN$. A vector $\x\in\Fpk$ is a codeword of the expander code iff for every $v\in \cAN\cup \cBN$, we have that $(x_{e_{v}(1)}x_{e_{v}(2)}\ldots x_{e_{v}(\Delta)}) $ is a codeword in $\cC_0$. Following~\cite{Zemorconcatenated}, we will call this the $(\cGN,\cC_0)$ code. The $(\cGN,\cC_0)$ code has blocklength $\Nout$ and dimension at least $\Nout\left( 1-2\left(\frac{\Delta-k_0}{\Delta}\right) \right)$~\cite{Zemorconcatenated}.

Z\'{e}mor~\cite{Zemorconcatenated} proposed an iterative algorithm for decoding expander codes. Suppose that the received (possibly erroneous) vector is $\y=(y_e:e\in \cEN)$. The vector $\widehat{\x}=(\widehat{x}_e:e\in \cEN)$ is initialized with $\widehat{x}_e=y_e$ for all $e\in \cEN$, and iteratively updated to obtain an estimate of $\x$. In every odd-numbered iteration, the algorithm replaces (for all $v\in \cAN$) $\widehat{\x}_v=(\widehat{x}_{e_{v}(1)}\widehat{x}_{e_{v}(2)}\ldots \widehat{x}_{e_{v}(\Delta)}) $ by a nearest-neighbour codeword (to $\widehat{\x}_v$) in $\cC_0$. In every even-numbered iteration, every $\widehat{\x}_v$ for $v\in \cBN$ is replaced by a nearest-neighbour codeword. This is repeated till $\widehat{\x}$ converges to a codeword in the expander code, or a suitably defined stopping point is reached. Since  $\{ E_v:v\in \cAN \}$ forms a disjoint partition of the edge set $\cEN$,  the nearest-neighbour decoding can be done in parallel for all the $v$'s in $\cAN$. The same holds for the vertices in $\cBN$. We direct the interested reader to~\cite{Zemorconcatenated} for more details about the code and the iterative decoding algorithm.

\begin{lemma}[\cite{Zemorconcatenated}]
   Let $\alpha<1$ be fixed. The iterative decoding algorithm can be implemented in a circuit of size $O(\Nout\log\Nout)$ and depth $O(\log\Nout)$ that always returns the correct codeword as long as the number of errors is less than $\frac{\alpha\delta_0\Nout}{2}\left( \frac{\delta_0}{2}-\frac{\lambda(\cGN)}{\Delta} \right)$.
   \label{lemma:Zemorconcatenated}
\end{lemma}
% Let us fix $\alpha=0.9$.
Since $\delta_0=4\sqrt{\epsilon}$ and $\frac{\lambda(\cGN)}{\Delta}\leq \frac{2\sqrt{\Delta-1}}{\Delta}\leq \sqrt{\epsilon}$, we see from Lemma~\ref{lemma:Zemorconcatenated} that the decoder can recover the transmitted outer codeword as long as the fraction of errors is less than $2\alpha\epsilon$. Although Lemma~\ref{lemma:Zemorconcatenated} was proved in~\cite{Zemorconcatenated} for binary expander codes, it can be verified that the result continues to hold in the case where the expander code is defined over $\Fpk$, provided that $p^k$ is a constant independent of $\Nout$.

\subsection{Performance of the Coding Scheme}

% We will show the following result.
\begin{theorem}
For every $\epsilon>0$, there exists a  sequence of concatenated codes $\cC^{(N)}$ with inner nested lattice codes and outer expander codes that satisfies the following for all sufficiently large $N$:
\begin{itemize}
\item rate $R^{(N)}\geq C-\epsilon$, 
\item  maximum transmit power $$\max_{\x\in \cCn}\frac{1}{N}\Vert \x \Vert^2\leq P-\epsilon,$$
\item the probability of error is at most $e^{-NE(\Rin)\epsilon}$, 
\item the  encoding complexity grows as $O(N^2)$, and
\item the decoding complexity grows as $O(N\log^2N)$.
\end{itemize}
\label{thm:main_expander}
\end{theorem}
\begin{proof}
% The outer expander code is guaranteed to correct all error patterns that have Hamming weight less than $W_{\mathrm{max}}=\frac{\alpha\delta_0\Nout}{2}\left( \frac{\delta_0}{2}-\frac{\lambda(\cGN)}{\Delta} \right)$. We fix $\alpha=0.9$, and it can be easily seen that $W_{\mathrm{max}}/\Nout>1.8\epsilon$.
Recall that the overall blocklength $N=n\Nout$, where $n$ is a sufficiently large constant.
The probability that an inner (lattice) codeword is recovered incorrectly is at most $\peinner$.
Let us fix $\alpha = 0.9$ and define $\delta_{\mathrm{out}}\triangleq  \frac{\alpha\delta_0}{2}\left( \frac{\delta_0}{2}-\frac{\lambda(\cGN)}{\Delta} \right)$, the fraction of errors that the outer expander code is guaranteed to correct according to Lemma~\ref{lemma:Zemorconcatenated}. From our choice of parameters, this quantity is at least $1.8\epsilon$.
The probability of error of the concatenated code can be upper bounded as follows:
\begin{align}
   P_{e,\mathrm{concat}}^{(N)}  &\leq  \begin{pmatrix}
 				  \Nout \\  \delta_{\mathrm{out}}\Nout+1	
                \end{pmatrix} \left(\peinner\right)^{\delta_{\mathrm{out}}\Nout+1} &\notag \\
                &\leq e^{\Nout \left(h\big(\delta_{\mathrm{out}}+1/\Nout\big)+o_{\Nout}(1)\right)} &\notag\\
                & \qquad\qquad \times e^{-nE(\Rin^{(n)})(\Nout \delta_{\mathrm{out}}+1)}&\notag \\
                &\leq e^{\Nout (h(\delta_{\mathrm{out}})+o_{\Nout}(1))}e^{-n\Nout\delta_{\mathrm{out}}E(\Rin^{(n)})}
\end{align}
For all sufficiently large $\Nout$, we can say that
\begin{align}
   P_{e,\mathrm{concat}}^{(N)}  &\leq \exp\left( -N\left( \delta_{\mathrm{out}}E(\Rin^{(n)})-\frac{(h(\delta_{\mathrm{out}})+\epsilon^2)}{n}  \right) \right) &\notag\\
                                                 &=\exp\left( -NE_{\mathrm{conc}}(\Rin^{(n)},\cGN,\cC_0) \right),
\end{align}
where the error exponent,
\begin{align}
   E_{\mathrm{conc}}(\Rin^{(n)},\cGN,\cC_0) &\triangleq  \delta_{\mathrm{out}}E(\Rin^{(n)})-\frac{(h(\delta_{\mathrm{out}})+\epsilon^2)}{n}. &\notag
\end{align}
Since $1.8\epsilon\leq \delta_{\mathrm{out}}<\delta_0=4\sqrt{\epsilon}$, we have
\begin{align}
E_{\mathrm{conc}}(\Rin^{(n)},\cGN,\cC_0)&\geq 1.8\epsilon E(\Rin^{(n)})-\frac{(h(4\sqrt{\epsilon})+\epsilon^2)}{n} &\notag\\
                      &= E(\Rin^{(n)})\left(1.8\epsilon -\frac{(h(4\sqrt{\epsilon})+\epsilon^2)}{nE(\Rin^{(n)})}\right) &\notag\\
                      &\geq  E(\Rin^{(n)})\epsilon
\end{align}
by our choice of $n$ in (\ref{eq:peinner_condition}).

Let us now inspect the encoding and decoding complexity. Recall that each floating-point operation has a complexity of $O(1)$. Since $n$ is a constant, encoding/decoding each inner (nested lattice) codeword requires  $O(1)$ floating-point operations, and there are $\Nout$ many codewords, leading to a total complexity of $O(\Nout)$. Since the outer code is linear, encoding requires $O(\Nout^2)$ operations in $\Fpk$. 
Since $p^{k}$ is a constant, the outer code has an encoding complexity of $O(\Nout^2) = O(N^2)$.
Decoding  the outer code requires $O(\Nout\log^2\Nout)$ operations in $\Fpk$.
We can therefore conclude that the decoding the concatenated code requires a complexity of $O(N\log^2 N)$, and encoding requires a complexity of $O(N^2)$.
This completes the proof of Theorem~\ref{thm:main_expander}.
\end{proof}

\begin{table*}
\begin{tabular}{|m{0.2\textwidth}|m{0.2\textwidth}|m{0.15\textwidth}|m{0.15\textwidth}|m{0.15\textwidth}|}
     \hline
     {\bf Scheme} & {\bf Decoding complexity ($\chi$)} & {\bf Encoding complexity}  & {\bf Error probability}  & {\bf Error probability as a function of $\chi$} \\\hline
     Polar lattice~\cite{YanPolar} &  $O(N\log^2N)$ &$O(N\log^2N)$ & $e^{-\Omega((N^\beta))},$   & $e^{-\Omega((\chi/\log^2 \chi)^{\beta})},$\\
                                                     &                        &                        & for any $0<\beta<\frac{1}{2}$         & for any $0<\beta<\frac{1}{2}$ \\\hline
     Sparse regression codes~\cite{JosephSuperposition} & $O(N^2)$ &$O(N^2)$ & $e^{-\Omega({N/\log N})}$ & $e^{-\Omega({\sqrt{\chi}/\log \chi})}$ \\\hline
     {RS-concatenated lattice codes}  & {$O(N^2)$}&{$O(N^2)$} & {$e^{-\Omega(N)}$} & {$e^{-\Omega(\sqrt{\chi})}$}\\\hline
     {Expander-concatenated lattice codes}  & {$O(N\log^2N)$} & {$O(N^2)$} &{$e^{-\Omega(N)}$} & {$e^{-\Omega(\sqrt{\chi})}$}\\\hline
     \end{tabular}
     \smallskip
     \caption{A comparison of the performance of various polynomial-time capacity-achieving codes.}
     \label{table:comparison}
  \end{table*}

%%%%%%%%%%%%%%%%%%%%%%%%%%%%%%%%%%%%%%%%%%%%%%%%%%%%%%%%%%%%%%%%%%%%%%%%
%%%%%%%%%%%%%%%%%%%%%%%%%%%%%%%%%%%%%%%%%%%%%%%%%%%%%%%%%%%%%%%%%%%%%%%%%
\section{Discussion}\label{sec:discussion}
The approach used in the previous sections can be used as a recipe for reducing the decoding complexity of optimal coding schemes for Gaussian channels.
A nested lattice scheme that achieves a rate $R$ over a Gaussian channel can be concatenated with a high-rate outer Reed-Solomon code or  expander code
to achieve any rate arbitrarily close to $R$. The only requirement is that the nested lattice code has a probability of error which decays exponentially in its blocklength. This procedure helps us bring down the decoding complexity to a polynomial function of the blocklength while ensuring that the probability of error continues to be an exponential function of the blocklength. 

\subsubsection{Gaussian Wiretap Channel}
As an application, consider the Gaussian wiretap channel~\cite{CheongWiretap}.
Tyagi and Vardy~\cite{TyagiExplicit} gave an explicit scheme using 2-universal hash functions that converts any coding scheme of rate $R$ over the point-to-point AWGN (main) channel to a coding scheme that achieves a rate $R-C_\mathrm{E}$ over the wiretap channel while satisfying the strong secrecy constraint. This ``conversion'' adds an additional decoding complexity which is polynomial in the blocklength. We can therefore use this result with Theorem~\ref{thm:main_RS} or Theorem~\ref{thm:main_expander} to conclude that we can achieve the secrecy capacity of the Gaussian wiretap channel with polynomial time decoding/encoding. 

\subsubsection{Compute-and-Forward}\label{sec:computeforward}
The compute-and-forward protocol was proposed by Nazer and Gastpar~\cite{Nazer11} for communication over Gaussian networks. Let us begin by describing the setup. We have $L$ source nodes $\tS_1,\tS_2,\ldots,\tS_L$, having independent messages $X_1,X_2,\ldots,X_L$ respectively. 
The messages are chosen from $\Fpk^{K}$ for some prime number $p$ and positive integers $k,K$. Let $\oplus$ denote the addition operator in $\Fpk^K$. These messages are mapped to $N$-dimensional real vectors $\u_1,\u_2,\ldots,\u_L$ respectively and transmitted across a Gaussian channel to a destination $\tD$ which observes
\begin{equation}
	\w = \sum_{l=1}^{L}h_l\u_l +\z, 
\label{eq:cforward_rec}
\end{equation}
where $h_1,h_2,\ldots,h_L$ are real-valued channel coefficients and $\z$ is AWGN with mean zero and variance $\nsvar$.  
The destination must compute $a_1X_1\oplus a_2X_2\oplus\cdots\oplus a_LX_L$, where $a_1,a_2,\ldots,a_L$ are integers.
We assume that each source node must satisfy a maximum power constraint of $P$.  We only consider symmetric rates here, i.e., all sources have identical message sets.
The rate of the code is $\frac{kK}{N}\log_2p$.
This problem is relevant in many applications such as exchange of messages in bidirectional relay networks, decoding messages over the Gaussian multiple access channel~\cite{Nazer11}, and designing good receivers for MIMO channels~\cite{ZhanIntegerForcing} to name a few. The basic idea is that instead of decoding the messages one at a time and using successive cancellation, it may be more efficient to decode multiple linear combinations of the messages. If we have $L$ linearly independent such combinations, then we can recover all the individual messages. 

We can extend the scheme of~\cite{Nazer11} to a concatenated coding scheme that achieves the rates guaranteed by~\cite{Nazer11}, but now with encoders and decoders that operate in polynomial time. Recall that the messages are chosen from $\Fpk^K$. We say that a rate $\mathcal{R}$ is achievable if for every $\epsilon>0$, there exists a sequence of encoders and decoders so that for all sufficiently large blocklengths $N$, we have the transmission rate $R^{(N)}\triangleq \frac{kK}{N}\log_2p >\mathcal{R}-\epsilon$, and the probability of error is  less than $\epsilon$.  We can show the following:
\begin{lemma}
Consider the problem of computing $a_1X_1\oplus a_2X_2\oplus\cdots\oplus a_LX_L$ from (\ref{eq:cforward_rec}).  Any rate 
\begin{equation}
   \cR < \frac{1}{2}\log_2\left( \frac{P}{\alpha^2 +P\sum_{l=1}^{L}(\alpha h_l-a_l)^2} \right), 
  \label{eq:cforward_rate}
  \end{equation}
  where 
  \begin{equation}
     \alpha\triangleq  \frac{P\sum_{l=1}^{L}h_la_l}{\nsvar + P\sum_{l=1}^{L}h_l^2},
    \label{eq:alphaval}
    \end{equation}
    is achievable with encoders and decoders whose complexities grow as $O(N^2)$ using an outer Reed-Solomon code, and a decoder whose complexity grows as $O(N\log^2N)$ with an outer expander code. For transmission rates less than $\cR$, the probability that the decoder makes an error goes to zero exponentially in $N$.
    \label{lemma:CF_lowcomplexity}
\end{lemma}
\begin{proof}
See Appendix~A.
 \end{proof}

 \subsubsection{Reducing the Probability of Error of LDA and Polar Lattices}
 The technique of concatenation can also be used to improve the error performance of other lattice codes that achieve the capacity of the AWGN channel.
 For example, polar lattices~\cite{YanPolar} have an error probability that decays as $e^{-\Omega(N^{\beta})}$ for any $0<\beta<1/2$, and LDA lattices~\cite{diPietro_arxiv} have an error probability that behaves as $O(1/N)$.
 It is easy to show that if $\peinner$ denotes the probability of error of the inner nested lattice code, then the probability of error of the (both RS and expander) concatenated code is
 \[
    P_{e}^{(N)}\leq \text{exp}\Big( -(\ln\peinner)\Nout\big((\peinner+\delta) - h(3\delta)/n\big) \Big)
 \]
 for some $\delta>0$. In any case, the probability of error goes to zero as $e^{-\Omega(\Nout)}$ irrespective of whether we use polar or LDA lattices. 
 We can therefore conclude that the probability of error decays as $e^{-\Omega(N/\log N)}$ for the corresponding Reed-Solomon concatenated (polar/LDA) lattice code and $e^{-\Theta(N)}$ for the corresponding expander concatenated (polar/LDA) lattice codes. The decoding complexities would grow as $O(N^2)$ (for RS concatenated codes) and $O(N\log^2N)$ (for expander concatenated codes) respectively.

\section{Concluding Remarks}\label{sec:concluding}
We have seen that concatenation can be a very powerful tool in reducing the asymptotic decoding complexity of nested lattice codes.
However, it  must be noted that achieving good performance using this scheme would require very large blocklengths. Although the 
probability of error decays exponentially in $N$, and the decoding/encoding complexities are polynomial in $N$, this is true only for very large values of $N$.
The fact that $N$ is at least exponential in the blocklength of the inner code is a major reason for this. 
Nevertheless, the concatenated coding approach shows that it is possible to obtain polynomial-time encoders and decoders for which the probability of error decays exponentially in the blocklength. The exponential decay is under the assumption that the gap between the transmission rate and capacity, $\gamma = C-R$, is kept fixed.
For a fixed error probability $P_e$, the blocklength required by the concatenated coding scheme to achieve rate $R=C-\gamma$ and error probability $P_e$ does not scale polynomially with $1/\gamma$. For a fixed error probability, we would like the complexity to not grow too fast as the rate approaches $C$. 
%Ideally, we want the gap to capacity $\gamma$ going to zero polynomially in $1/N$. 
It has been recently shown that polar codes have this property for binary memoryless symmetric channels~\cite{ModelliScaling}. Designing codes for the Gaussian channel whose decoding/encoding complexities are also polynomial in $1/\gamma$ for a fixed probability of error still remains an open problem.

\section{Acknowledgements}
The authors would like to thank Prof.\ Sidharth Jaggi for a discussion that led to this work. The work of the first author was supported by the Tata Consultancy Services Research Scholarship Program, and that of the second author by a Swarnajayanti Fellowship awarded by the Dept.\ of Science and Technology (DST), Govt.\ of India.

\section*{Appendix~A: Proof of Lemma~\ref{lemma:CF_lowcomplexity}}

The technique used to prove Lemma~\ref{lemma:CF_lowcomplexity} is a simple extension of the coding scheme of~\cite{Nazer11} using the methods described in Section~\ref{sec:codingscheme}. For completeness, we will briefly describe the scheme. For more details regarding the compute-and-forward protocol, see~\cite{Nazer11}.  
We use the concatenated coding scheme of Section~\ref{sec:concat_scheme_awgn}.
The inner code is obtained from nested Construction-A lattices $(\Lfn,\Lcn)$. 
Suppose that $\Lfn$ is constructed using a $(n,k)$ linear code over $\bFp$.
The outer code is an $(\Nout,\Kout,\Nout-\Kout+1)$ Reed-Solomon code, with $\Nout = p^k-1$ and $\Kout$ to be specified later. The transmission rate is $R^{(n)}=\frac{k\Kout}{n\Nout}\log_2 p$.

The messages are chosen from $\Fpk^{\Kout}$. Let the message at the $l$th user be $M_l=[m^{(l)}_1,m^{(l)}_2,\ldots,m^{(l)}_{\Kout}]^{T}$, where $m^{(l)}_i\in \Fpk$.  The messages are mapped to an $\Nout$-length codeword over $\Fpk$ using the outer code. Let the resulting codeword be $\y^{(l)}=[y^{(l)}_1,y^{(l)}_2,\ldots,y^{(l)}_{\Nout}]^{T}$.

 Each $y^{(l)}_i$ is then encoded to $\u^{(l)}_i$ using the inner code and then transmitted. Recall that there exists a group isomorphism from $\Lfn/\Lcn$ to  $\Fpk$. For $1\leq l\leq L$ and $1\leq i\leq \Nout$, let $\x^{(l)}_i$ be the representative of $y^{(l)}_i$ in $\Lfn\cap \cV(\Lcn)$. Independent dither vectors $\t^{(l)}_1,\t^{(l)}_2,\ldots,\t^{(l)}_{\Nout}$ are generated at the $L$ sources. Transmitter $l$ successively sends $\u^{(l)}_i=[\x^{(l)}_i-\t^{(l)}_i]\bmod\Lcn$ for $1\leq i\leq\Nout$ to the receiver. 

  The decoder, upon receiving $\w_i=\sum_{l=1}^{L}\u^{(l)}_i+\z$, computes $\widetilde{\w}_i=\big[ \alpha \w_i +\sum_{l=1}^{L}a_l\t_i^{(l)} \big]\bmod\Lcn$.
  The estimate of  $[\sum_{l=1}^{L}a_{i}\x_i^{(l)}]\bmod\Lcn$, is $[Q_{\Lfn}(\widetilde{\w}_i)]\bmod\Lcn$.
  Recall the definition of $\cR$ in (\ref{eq:cforward_rate}).
  Nazer and Gastpar showed in~\cite{Nazer11} that there exists a sequence of nested Construction-A lattices with $\Rin^{(n)}=\frac{k}{n}\log_2p$ for which the probability that the decoder makes an error in estimating the desired linear combination decays as $e^{-nE_c(\Rin^{(n)})}$, where $E_c(\cdot)$ is some function which is positive for all $\Rin^{(n)}<\cR$. As we did before for the AWGN channel, we choose $\Kout=\Nout(1-2e^{-nE_c(\Rin^{(n)})}-\epsilon)$. 
  Assuming that fewer than $(\Nout-\Kout)/2$ inner codewords are in error, the decoder can recover  $\widehat{\x}_c = \Big[\big[\sum_{l}a_l\x_1^{(l)}\big]\bmod\Lcn,\ldots,\big[\sum_{l}a_l\x_{\Nout}^{(l)}\big]\bmod\Lcn\Big]^T$ without error. Due to the existence of a group isomorphism between $\Fpk$ and $\Lfn/\Lcn$, this implies that the decoder can recover $a_1\y^{(1)}\oplus\cdots\oplus a_L\y^{(L)}$, and hence, $a_1M_1\oplus\cdots\oplus a_LM_L$. Arguing as in Section~\ref{sec:codingscheme}, the probability  that the decoder makes an error goes to zero exponentially in $N$, and the decoding/encoding complexities grow as $O(N^2)$. The same arguments can be used to show that using an outer expander code, we can have the encoding complexity to be $O(N^2)$ and decoding complexity to be $O(N\log^2N)$.
%   \qed

\end{document}